%% file: diameter-arxiv.tex
\documentclass[11pt]{article} 
\usepackage[utf8]{inputenc} 
\usepackage{amsmath,amsthm,amssymb,geometry,verbatim} 
\usepackage[colorlinks=true,linkcolor=red,citecolor=blue,urlcolor=red]{hyperref}
\usepackage{tablefootnote}
\usepackage{paralist}
\usepackage{fullpage}
\usepackage{framed}
\usepackage{mathptmx}
\usepackage{array} 
\usepackage{fancyhdr} 
\usepackage{sectsty}

\def \eps {\varepsilon}

\def \R {{\mathbb R}}
\def \Z {{\mathbb Z}}

\newcommand{\ip}[2]{\ensuremath{\left<#1,#2\right>}}
\newcommand{\ov}{\text{\bf OV}}

\newtheorem{theorem}{Theorem}[section]
\newtheorem{corollary}{Corollary}[section]
\newtheorem{lemma}{Lemma}[section]
\newtheorem{conjecture}{Conjecture}[section]

\def \poly {\text{poly}}

\newenvironment{reminder}[1]{\bigskip
\noindent {\bf Reminder of #1  }\em}{\smallskip}

\title{On the Difference Between Closest, Furthest, and Orthogonal Pairs: \\ {\normalsize Nearly-Linear vs Barely-Subquadratic Complexity in Computational Geometry}}
\author{Ryan Williams\thanks{MIT CSAIL \& EECS, {\tt rrw@mit.edu}. Supported by an NSF CAREER award.}}
\begin{document}
\date{}

\maketitle
\begin{abstract} 

Point location problems for $n$ points in $d$-dimensional Euclidean space (and $\ell_p$ spaces more generally) have typically had two kinds of running-time solutions: 

\begin{itemize}
\item[(Nearly-Linear)] less than $d^{\poly(d)} \cdot n \log^{O(d)} n$ time, or 
\item[(Barely-Subquadratic)] $f(d) \cdot n^{2-1/\Theta(d)}$ time, for various functions $f$. 
\end{itemize}

For small $d$ and large $n$, ``nearly-linear'' running times are generally feasible, while the ``barely-subquadratic'' times are generally infeasible, requiring essentially quadratic time. For example, in the Euclidean metric, finding a Closest Pair among $n$ points in $\R^d$ is nearly-linear, solvable in $2^{O(d)} \cdot n \log^{O(1)} n$ time, while the known algorithms for finding a Furthest Pair (the diameter of the point set) are only barely-subquadratic, requiring $\Omega(n^{2-1/\Theta(d)})$ time. Why do these proximity problems have such different time complexities? Is there a barrier to obtaining nearly-linear algorithms for problems which are currently only barely-subquadratic?

We give a novel exact and deterministic self-reduction for the Orthogonal Vectors problem on $n$ vectors in $\{0,1\}^d$ to $n$ vectors in $\Z^{\omega(\log d)}$ that runs in $2^{o(d)}$ time. As a consequence, barely-subquadratic problems such as Euclidean diameter, Euclidean bichromatic closest pair, ray shooting, and incidence detection do not have $O(n^{2-\epsilon})$ time algorithms (in Turing models of computation) for dimensionality $d = \omega(\log \log n)^2$, unless the popular Orthogonal Vectors Conjecture and the Strong Exponential Time Hypothesis are false. That is, while the poly-log-log-dimensional case of Closest Pair is solvable in $n^{1+o(1)}$ time, the poly-log-log-dimensional case of Furthest Pair can encode difficult large-dimensional problems conjectured to require $n^{2-o(1)}$ time.

We also show that the All-Nearest Neighbors problem in $\omega(\log n)$ dimensions requires $n^{2-o(1)}$ time to solve, assuming either of the above conjectures.  

\end{abstract}

\thispagestyle{empty}
\newpage
\setcounter{page}{1}

\section{Introduction} 

Point proximity and location problems have been core to computer science and computational geometry since Minsky and Papert~\cite{Minsky-Papert69} and Knuth's post office problem~\cite{knuth1973art}. In this paper, we study the problems of finding the closest pair or furthest pair in a point set (i.e., the \emph{diameter}) in moderate dimensions under the most natural norms, and incidence problems such as Hopcroft's problem~\cite{Matousek93,Erickson95,Erickson96}: given $n$ points in $\R^d$ and $n$ hyperplanes through the origin, does any point lie on any line? (Note this is equivalent to asking whether there are two vectors which are orthogonal, i.e., have inner product $0$.)  For closest and furthest pair problems, we also consider their \emph{bichromatic} versions where there are $n$ red points, $n$ blue points, and we wish to find a closest (or furthest) red/blue pair.
\footnote{Note we do not consider $\ell_1$ and $\ell_2$ bichromatic furthest pair explicitly, since it is easy to efficiently reduce between the bichromatic version and the uncolored version. For example, we can reduce from bichromatic to non-bichromatic by adding one extra dimension with large (positive if red, negative if blue) coordinates.}
 We consider these problems under the $\ell_p$ metric for $p \in \{1,2\}$, as well as $\ell_{\infty}$. As is standard, we use $\ell_p^d$ to denote the metric space $(\R^d,\ell_p)$, with the distance functions $||x-y||_p = (\sum_{i=1}^d |x_i - y_i|^p)^{1/p}$ and $||x-y||_{\infty} = \max_{i} |x_i - y_i|$.

For the case of very large $n$ and modest $d$, some of these problems appear to be far more difficult to solve than others, for reasons which are still not well-understood (beyond the fact that known techniques do not work). As early as 1976, Bentley and Shamos~\cite{bentley1976divide} noticed an apparent difference in the difficulties of solving furthest pair and closest pair in $\ell_2$ in higher dimensions, and raised it as an important issue to study. The following table gives a rough classification of key problems which are known to be ``easy'' and which seem to be ``hard'' for large $n$ and modest $d$.\footnote{In this paper, we assume a machine model that allows basic arithmetic on entries of vectors and comparisons of points in $\Z^d$ in $\poly(d,\log M)$ time, where $M$ is the largest magnitude of an integer in an input. Such a concrete model is necessary for our hardness results, which are concerned with discrete tasks such as SAT-solving in typical Turing models of computation.}

\begin{center}
\begin{tabular}{l||l}
{\bf Nearly-Linear}   			($d^{\poly(d)} \cdot n \log^{O(d)} n$ time) &		 {\bf Barely-Subquadratic}   ($f(d) \cdot n^{2-1/\Theta(d)}$ time) \\
\hline
(Bichrom.) $\ell_{\infty}^d$-Furthest Pair~\cite{Yao82,Gabow-Bentley-Tarjan}  &   
$\ell_2^d$-Furthest Pair~\cite{Yao82,Agarwal1990euclidean} \\
$\ell_2^d$-Closest Pair~\cite{bentley1976divide,khuller1995simple,dietzfelbinger1997reliable}
&   Bichrom.~$\ell_2^d$-Closest Pair~\cite{Agarwal1990euclidean}\\
$\ell_1^d$-Furthest Pair~\cite{Yao82,Gabow-Bentley-Tarjan}	
&		$d$-dim. Hopcroft's Problem~\cite{Chazelle93,Matousek93}\\
(Bichrom.) $\ell_1^d$ and $\ell_{\infty}^d$-Closest Pair \\
~~~~~~\cite{Gabow-Bentley-Tarjan,Preparata1985introduction,dietzfelbinger1997reliable,Chan17}	
\end{tabular}
\end{center}

Note that there are many other core geometry problems with one of the two above runtime types; the above are just some of the core bottlenecks. For example, Hopcroft's problem is a special case of problems such as (batch) point location and ray shooting, which also suffer from the same $n^{2-1/\Theta(d)}$ dependency~(see Erickson's work on hardness from Hopcroft's problem~\cite{Erickson95} for more).

Why do some problems fall on the right side of the table, and can they be moved to the left side? Besides the natural question of understanding the difference between furthest and closest pair, here is another motivating example. In 1984, Gabow, Bentley, and Tarjan~\cite{Gabow-Bentley-Tarjan} showed that the $\ell_{\infty}$-furthest pair problem (and its bichromatic version) in $\R^d$ is \emph{very} easy, solvable in $\tilde{O}(d \cdot n)$ time. Using this fast algorithm, along with an isometric embedding of $\ell_1^d$ into $\ell_{\infty}^{2^d}$, they then solve the (bichromatic or not) furthest pair problem for $\ell_1$ in $\tilde{O}(2^d \cdot n)$ time. 
So computing the $\ell_{\infty}$-diameter and $\ell_1$-diameter are both ``nearly-linear'' time problems in low dimensions. 

{\bf Can similar bounds be achieved for $\ell_2$-furthest pair?} As the above table indicates, the best known algorithms for furthest pair in $\ell_2$ (bichromatic or not) still have running time bounds of the form $O(n^{2-1/\Theta(d)})$, which is ``barely subquadratic.'' Is there a fundamental reason why this problem is so much harder in $\ell_2$ than in $\ell_1$ or in $\ell_{\infty}$? 

The situation is arguably counter-intuitive, because $\ell_1$ and $\ell_{\infty}$ are technically more ``universal'' metrics than $\ell_2$, so one might think that problems should be more difficult under the former than the latter. For instance, efficient isometric embeddings of $n$-point sets from $\ell_2$ into $\ell_1$ and into $\ell_{\infty}$ \emph{are} known in the literature on metric embeddings (see the book~\cite{Deza-Laurent97} for references), whereas the converse is not true (see for example~\cite[Chapter 2]{Wells-Williams75}). However, these isometric embeddings need $\Omega(n)$ dimensions in the most general cases. There may still be embeddings (perhaps randomized) which map low-dimensional $n$-point sets in $\ell_2$ into sub-exponential-dimensional $\ell_1$ (or $\ell_{\infty}$). Indeed, in the case of low distortion (where the distances in an embedding are allowed to shrink or grow by small multiplicative amounts) these are well-known, even deterministically in some regimes~\cite{linial1994geometry,indyk2007uncertainty,guruswami2010almost}. The results of this paper show that ``nice'' isometric embeddings of $\ell_2$ into $\ell_1$ would have major implications in fine-grained complexity.

\subsection{Strong Difficulty of Proximity Problems in the Euclidean metric}

We offer good reasons why furthest pair in $\ell_2$ and other barely-subquadratic problems will be difficult to solve as fast as closest pair, even in very low dimensions. We do this by relating $\ell_2$-furthest pair and other ``barely subquadratic'' problems to the Orthogonal Vectors Conjecture~\cite{Williams05,DBLP:conf/icalp/AbboudVW14} and the Strong Exponential Time Hypothesis~\cite{IP01,CIP09} in a novel way. 

The \ov{} problem is: \emph{given $n$ vectors $v_1,\ldots,v_n \in \{0,1\}^d$, are there $i,j$ such that $\langle v_i,v_j\rangle = 0$?} Clearly $O(n^2 d)$ time suffices for solving \ov{}, and slightly subquadratic-time algorithms are known in the case of small $d$~\cite{AbboudWY15,DBLP:conf/soda/ChanW16}. It is conjectured that there is no \ov{} algorithm running in (say) $n^{1.99}$ time for dimensionality $d = \omega(\log n)$.

\begin{conjecture}[Orthogonal Vectors Conjecture (OVC) \cite{Williams05,DBLP:conf/icalp/AbboudVW14}]
For every $\eps > 0$, there is a $c \geq 1$ such that \ov{} cannot be solved in $n^{2-\eps}$ time on instances with $d = c\log n$. 
\end{conjecture}

In other words, OVC states that \ov{} requires $n^{2-o(1)}$ time on instances of dimension $\omega(\log n)$. OVC is plausible because it is implied by (and looks much more likely than) the popular Strong Exponential Time Hypothesis~\cite{IP01,CIP09} on the time complexity of solving $k$-SAT~\cite{Williams05,Williams-Yu14}. 

Straightforward transformations show that OVC implies that both furthest and bichromatic closest pair in $\ell_1^{\omega(\log n)}$ and $\ell_2^{\omega(\log n)}$ require $n^{2-o(1)}$ time~\cite{Williams05,AlmanW15}. Also assuming OVC, David, Kartik, and Laekhanukit~\cite{David16} show that (non-bichromatic) closest pair in $\ell_p^{\omega(\log n)}$ for $p > 2$ and $\ell_{\infty}^{\omega(\log n)}$ also require $n^{2-o(1)}$ time. It is not so surprising that some proximity search problems in super-log dimensions are hard under OVC, because OVC is a hardness conjecture about a problem in super-log dimensions. 

In this paper, we show that OVC implies bichromatic closest pair and furthest pair in $\ell_2$ require essentially quadratic time for even \emph{poly-loglog} dimensions, in stark contrast with bichromatic closest pair and furthest pair in both $\ell_1$ and $\ell_{\infty}$ (which both have $n^{1+o(1)}$-time solutions in this case). Our main technical tool is the following dimensionality reduction for Orthogonal Vectors:

\begin{lemma}[Dimensionality Reduction for OV] \label{dim-red} Let $\ell \in [1,d]$. There is an $n \cdot d^{O(d/\ell)}$-time reduction from \ov{} for $n$ points in $\{0,1\}^d$ to $d^{O(d/\ell)}$ instances of \ov{} for $n$ points in $\Z^{\ell+1}$, with vectors of $O((d \log d)/\ell)$-bit entries.
\end{lemma}

Applying this lemma, we establish quadratic-time hardness for the barely-subquadratic Hopcroft's problem, $\ell_2$-Furthest Pair, and Bichromatic $\ell_2$-Closest Pair in small (poly-log-log) dimensions. It follows that if any one of these three problems became ``nearly-linear'', then there would have many interesting algorithmic consequences, including new SAT-solving algorithms. For example:

\begin{theorem}[Hardness of Hopcroft's Problem] \label{hopcroft} Under SETH (or OVC), Hopcroft's problem in $\omega(\log \log n)$ dimensions requires $n^{2-o(1)}$ time, with vectors of $O(\log n)$-bit entries.
\end{theorem}

\begin{theorem}[Hardness of $\ell_2$-Furthest Pair] \label{furthest-pair} Under SETH (or OVC), finding a furthest pair in $\omega(\log \log n)^2$ dimensions under the $\ell_2$ norm requires $n^{2-o(1)}$ time, with vectors of $O(\log n)$-bit entries.
\end{theorem}

Therefore, computing the diameter of an $n$-point set in low-dimensional $\ell_2$ is surprisingly more difficult to solve than in the $\ell_1$ metric, or in the $\ell_{\infty}$ metric. By Gabow-Bentley-Tarjan~\cite{Gabow-Bentley-Tarjan}, there are $n^{2-\eps}$-time algorithms for furthest pair under $\ell_1$ up to ${\eps \log n}$ dimensions, and under $\ell_{\infty}$ up to $n^{1-\eps}$ dimensions. There seems to be an exponential curse of dimensionality in computing the diameter of a point set, going from $\ell_{\infty}$ to $\ell_1$, and \emph{also} going from $\ell_1$ to $\ell_2$. The following table summarizes the consequences for barely-subquadratic problems.

\begin{center}
\begin{tabular}{l|l}
 {\bf Barely-Subquadratic Problem} & {\bf Lower Bound (Under SETH or OVC)}\\
\hline
$\ell_2^d$-Furthest Pair~\cite{Yao82,Agarwal1990euclidean} & $n^{2-o(1)}$ time for $d=\omega(\log \log n)^2$ \\
 Bichrom.~$\ell_2^d$-Closest Pair~\cite{Agarwal1990euclidean} 
 & $n^{2-o(1)}$ time for $d=\omega(\log \log n)^2$
 \\
	$d$-dim. Hopcroft's Problem~\cite{Chazelle93,Matousek93}
	 & $n^{2-o(1)}$ time for $d=\omega(\log \log n)$
	\\
\end{tabular}
\end{center}

Under the present landscape of fine-grained complexity conjectures, it follows that none of the barely-subquadratic problems we have identified can be made nearly-linear:

\begin{corollary} Under SETH (or OVC), {\bf none} of $\ell_2$-Furthest Pair, Bichromatic $\ell_2$-Closest Pair, or Hopcroft's problem are solvable in $n^{2-\eps} \cdot \log^{2^{o(\sqrt{d})}} n$ time, for all $\eps > 0$.
\end{corollary}

Since the above barely-subquadratic problems have closely-related nearly-linear problems, these results also show that OVC and SETH have consequences for the theory of metric embeddings. For example, since $\ell^d_{\infty}$-Furthest Pair can be solved in $\tilde{O}(d \cdot n)$ time, every $n^{1.99}$-time isometric embedding from $n$ points in $\ell^d_2$ into $\ell_{\infty}$ with $d=\omega(\log \log n)^2$ must blow up the dimension doubly-exponentially to $n^{1-o(1)}$ --- unless OVC and SETH are false. This is striking when one remembers that \emph{every $n$-point metric} can be (efficiently) isometrically embedded into $\ell_{\infty}$ with $n-1$ dimensions (by the classical Frechet embedding). 

Unfortunately, the above conditional lower bounds only hold for exact solutions to the problems. Our reductions from \ov{} to closest/furthest pair no longer work if we only have $(1+\eps)$-approximations to the closest/furthest pair (if they did, this paper would be about how OVC is false, thanks to many fast approximation algorithms for these problems~\cite{AndoniIndyk17}). 

\paragraph{Hardness for All-Nearest Neighbors.} The best known algorithms for the $\ell_2$-Closest Pair problem are nearly-linear, running in $2^{O(d)} n \log^{O(1)}n$ time. A prominent open problem is whether the exponential dependence on $d$ is necessary: \emph{Does $\ell_2$-Closest Pair require $n^{2-o(1)}$ time in $\omega(\log n)$ dimensions?} Could we show hardness under (for example) OVC or SETH?

The question is rather subtle. As mentioned earlier, the related problems of Bichromatic $\ell_2$-Closest Pair and $\ell_2$-Furthest Pair are easily shown to be \ov{}-hard in $\omega(\log n)$ dimensions~\cite{Williams05,AlmanW15}. Intuitively speaking, in both of the latter problems, our reductions can ``control'' the distances between points in such a way that it is easy to encode \ov{}. But for $\ell_2$-Closest Pair (with no colors), we have much less control, and it is difficult to keep large sets of points far enough apart to successfully encode an \ov{} instance~\cite{David16}.

Here we report some progress on this open problem. In the closely-related \emph{All-Nearest Neighbors} problem, the task is to report the $\ell_2$-closest pair for all points in the given set. Nearly-linear algorithms are also known for All-Nearest Neighbors, which have essentially the same complexity as $\ell_2$-Closest Pair~\cite{Clarkson83,Vaidya89}. We can show \ov{}-hardness for All-Nearest Neighbors:

\begin{theorem} \label{all-nn} Under OVC, the All-Nearest Neighbors problem in $\ell^{\omega(\log n)}_2$ requires $n^{2-o(1)}$ time, even restricted to vectors with entries from $\{-1,0,1\}$.
\end{theorem}

The reduction goes through the Set Containment problem (equivalent to \ov{}), and uses error-correcting codes to keep one half of the vectors ``distant'' from each other, and the other half relatively ``close'' to the first half.

\section{A Dimensionality Self-Reduction for Orthogonal Vectors}

In this section, we set up the framework for proving hardness for the aforementioned ``barely-subquadratic'' problems. We begin with the following more general theorem, which will imply the dimensionality reduction lemma.

\begin{theorem}\label{REDDIM}
For every $d$ and integer $\ell \in [1,d]$, given two sets of vectors $U,V\subseteq\{0,1\}^d$, there is a deterministic algorithm running in $n \cdot d^{O(d/\ell)}$ time which outputs a list of $t = d^{O(d/\ell)}$ integers $\{k_1,\ldots,k_t\} \subseteq [0,t]$, along with sets $U',V'\subseteq\Z^{\ell+1}$ such that $|U'| = |U|$, $|V'|=|V|$, and all entries in $u',v'$ are $O((d\log d)/\ell)$-bit integers. There is an orthogonal pair $u\in U,v\in V$ if and only if there is a pair $u'\in U',v'\in V'$ such that $\langle u',v'\rangle = k_i$ for some $i$. 
\end{theorem}

Although it may be difficult to see in hindsight, the proof of Theorem~\ref{REDDIM} is inspired by the Merlin-Arthur communication protocol with $\tilde{O}(\sqrt{d})$ communication for Inner Product, due to Aaronson and Wigderson~\cite{Aaronson-Wigderson09}. In that protocol, two parties each hold a $d$-bit vector, and they wish to determine if their vectors are orthogonal. The protocol shows how a prover can send an $\tilde{O}(\sqrt{d})$-bit message to the two parties, such that the two parties only need to exchange $\tilde{O}(\sqrt{d})$-bits (with $O(\log d)$ public randomness) to determine orthogonality with high probability. They do this by encoding $d$-bit vectors with $O(\sqrt{d})$-degree bivariate polynomials, and their protocol uses the key good property of low-degree polynomials that we know (they have few roots, so evaluating distinct two polynomials at a random point will yield two distinct values, with decent probability). 

In the below proof of Theorem~\ref{REDDIM}, there are several major differences. First, we forget one of the variables, and encode our $d$-bit vectors with $\ell$-dimensional vectors whose entries are $d/\ell$-degree univariate polynomials. Second, the parameter $\ell$ allows for a trade-off between the length of the vector and the degrees of the polynomials. (This corresponds to a trade-off between the length of the prover's message and the length of the others' messages, in the Merlin-Arthur protocol.) Third, we do \emph{not} pick random points to evaluate the polynomial on, but rather a single deterministic value. This actually suffices for our purposes.

\begin{proof} Without loss of generality, assume $d$ is a multiple of $\ell$, otherwise we can add zeroes at the end of each vector to satisfy this assumption.

Consider two vectors $u \in U,v \in V$. Divide the $d$ dimensions of both vectors into $\ell$ contiguous blocks, each of which contains $d/\ell$ dimensions. 
Suppose the $i$th block of $u$ is $[u_{i,1},\ldots,u_{i,d/\ell}]$ and the $i$th block of $v$ is $[v_{i,1},\ldots,v_{i,d/\ell}]$, where all $u_{i,j}, v_{i,j} \in \{0,1\}$. Construct the polynomials 
\[P_{u,i}(x) =\sum_{j=1}^{d/\ell} u_{i,j}\cdot x^{j-1}\] and 
\[Q_{v,i}(x) = \sum_{j=1}^{d/\ell} v_{i,j} \cdot x^{d/\ell-j}.\] 
Let $P_u(x)$ be the $\ell$-dimensional vector $[P_{u,1},\ldots,P_{u,\ell}]$ and 
$Q_v(x)$ be the $\ell$-dimensional vector $[Q_{v,1},\ldots,Q_{v,\ell}]$. Observe that the coefficient of $x^{d/\ell-1}$ in the polynomial $R_{u,v}(x) = \ip{P_u(x)}{Q_v(x)}$ is exactly
\[\sum_{i=1}^{\ell} \sum_{j=1}^{d/\ell} u_{i,j}\cdot v_{i,j} = \ip{u}{v}.\] Furthermore, note that for any $u \in U$ and $v \in V$, the polynomial $R_{u,v}(x)$ has degree at most $2d/\ell$, and each of its coefficients are integers in $[0,d]$. 

Now we are ready to describe the reduction. First, enumerate all $t = d^{O(d/\ell)}$ polynomials $R(x)$ of degree at most $2d/\ell$ with coefficients in $[0,d] \cap \Z$ such that the coefficient of $x^{d/\ell-1}$ equals $0$. 

Set $x_0 := d+1$. Note that, given the integer value $k = R(x_0) = \sum_{i=0}^{\ell-1}P_{u,i}(x_0) \cdot Q_{v,i}(x_0)$, the polynomial $R(x)$ is uniquely determined, because all of its coefficients are integers in $[0,d]$.

For all $u \in U$ and $v \in V$, compute $u' := P_u(x_0)$ and $v' := Q_v(x_0)$, creating two sets of vectors $U'$ and $V'$ where all vectors have $\ell$ dimensions, with entries of bit length at most $O((d \log d)/\ell)$. 

By enumerating over all such polynomials $R(x)$, we obtain sets $U'$, $V'$, and collection of $t$ integers $\{R(x_0)\}$ each in $O(d\log d)/\ell$ bits, satisfying the conclusion of the theorem. In particular, the vectors $u \in U$, $v \in V$ satisfy $\ip{u}{v} = 0$ if and only if there is \emph{some} polynomial $R(x)$ of degree at most $2d/\ell$ with coefficients in $[0,d] \cap \Z$ such that $\ip{P_u(x_0)}{Q_v(x_0)} = R(x_0)$. \end{proof}

Now we prove the dimensionality reduction lemma:

\begin{reminder}{Lemma~\ref{dim-red}}[Dimensionality Reduction for OV] Let $\ell \in [1,d]$. There is an $n \cdot d^{O(d/\ell)}$-time reduction from \ov{} for $n$ points in $\{0,1\}^d$ to $d^{O(d/\ell)}$ instances of \ov{} for $n$ points in $\Z^{\ell+1}$, with vectors of $O(\log n)$-bit entries.
\end{reminder}

\begin{proof} Given a set $S$ of $n$ (non-zero) vectors in $\{0,1\}^d$, set $U := S$ and $V := S$ in Theorem~\ref{REDDIM}, which produces $n$ vectors $U'$ and $V'$ in $\Z^{\ell}$ along with a set of $d^{O(d/\ell)}$ numbers $T$ such that $S$ has an orthogonal pair if and only if there is some $u \in U'$, $v \in V'$, and $k \in T$ such that $\ip{u}{v}=k$. 

For every $k \in T$, create new sets of vectors $U'_k,V'_k$, where every $u \in U'$ is  replaced by $u_k := [u, 1]$ in $U'_k$, and every $v \in V'$ is replaced by $v_k := [v, -k]$ in $V'_k$. Since all entries in $u$ and $v$ are non-negative, we observe:
\begin{enumerate}
\item for all $u \in U'$ and $v \in V'$, $\ip{u}{v}=k$ if and only if $\ip{u_k}{v_k}=0$, 
\item for every pair $u_k, u'_k \in U_k$,  $\ip{u_k}{u'_k}\geq 1$, and
\item for every pair $v_k, v'_k \in V_k$,  $\ip{v_k}{v'_k}\geq k^2$.
\end{enumerate}
Consider the set $S_k := U_k \cup V_k \subset \Z^{\ell}$. By the above three facts, we could only obtain an orthogonal pair of vectors in $S_k$ by taking one vector from $U_k$ and one vector from $V_k$, and $S_k$ contains an orthogonal pair if and only if there is some $u \in U'$ and $v \in V'$ such that $\ip{u}{v}=k$. Our reduction calls $\ov{}$ on $S_k$ for every $k \in T$, and outputs the relevant orthogonal pair for $S$ if any of the calls return an orthogonal pair for some $S_k$.
\end{proof}

\subsection{Consequences}

Here we show how the above Dimensionality Reduction for OV implies hardness for the barely-subquadratic problems mentioned in the introduction.

\begin{reminder}{Theorem~\ref{hopcroft}}[Hardness of Hopcroft's Problem]  Under SETH (or OVC), Hopcroft's problem in $\omega(\log \log n)$ dimensions requires $n^{2-o(1)}$ time, with vectors of $O(\log n)$-bit entries.
\end{reminder}

\begin{proof} Let $c \geq 1$ be an arbitrary constant and let $d := c\log n$. We show how an oracle for Hopcroft's problem in $\omega(\log \log n)$ dimensions, running in $O(n^{2-\delta})$ time for some universal $\delta > 0$, can be used to solve \ov{} for $n$ vectors in $d$ dimensions in $n^{2-\delta+\eps}$ time (regardless of $c$) for every $\eps > 0$, which would refute the OVC.

Set $\ell := c(\log d)/\alpha = c \log(c \log n)/\alpha$, for a small parameter $\alpha > 0$ to be set later. Applying Lemma~\ref{dim-red} to a given subset $S \subseteq \{0,1\}^d$, the reduction runs in time \[n \cdot (c \log n)^{O(c\log n)/\ell)} \leq n \cdot c^{O(\alpha \log n)} \leq n^{1+O(\alpha \log(c))},\]and produces $n^{O(\alpha \log(c))}$ instances of \ov{} with $n$ points in $\Z^{(c \log \log n)/\alpha + O(1)}$, with vectors of $O(\log n)$-bit entries. Setting $\alpha \ll \eps/\log(c)$, the reduction generates $O(n^{\eps})$ instances in $\Omega(1/\eps \cdot c \log(c) \cdot \log \log n)$ dimensions, each of which our Hopcroft oracle solves in $n^{2-\delta}$ time, by assumption. This concludes the proof.
\end{proof}

\begin{reminder}{Theorem~\ref{furthest-pair}}[Hardness of $\ell_2$-Furthest Pair]  Under SETH (or OVC), finding a $\ell_2$-furthest pair in $\omega(\log \log n)^2$ dimensions requires $n^{2-o(1)}$ time, with vectors of $O(\log n)$-bit entries.
\end{reminder}

\begin{proof} Given a fast algorithm for $\ell_2$-furthest pair in $\omega(\log \log n)^2$ dimensions, we show how to quickly solve Hopcroft's problem on $n$ points in $\omega(\log \log n)$ dimensions, and appeal to Theorem~\ref{hopcroft}.

Let $S$ be a set of $n$ vectors in $\Z^{\ell}$ with $\ell = \omega(\log \log n)$ and with $O(\log n)$-bit entries. Let $k > 1$ be such that every entry of every vector has magnitude less than $n^k$. In the following, let $v[i]$ denote the $i$th component of a vector $v$.

For every vector $u \in S$, define the $(\ell^2+2)$-dimensional vector
\[u' := [u[1]\cdot u[1], u[1]\cdot u[2], \ldots, u[i]\cdot u[j], \ldots, u[\ell]\cdot u[\ell],0, n^{2k+1}].\] That is, the first $\ell^2$ components of $u'$ are all possible products of two components of $u$, followed by the entries $0$ and $n^{2k+1}$. Put each $u'$ in a set $U'$. Also for every vector $v \in S$, define the $(\ell^2+2)$-dimensional vector
\[v' := [v[1]\cdot v[1], v[1] \cdot v[2], \ldots, v[i]\cdot v[j], \ldots, v[\ell]\cdot v[\ell],n^{2k+1},0],\] and put $v'$ in a set $V'$. Now observe that:
\begin{itemize}
\item for $u'_1,u'_2 \in U'$ coming from some $u_1,u_2 \in S$, $\ip{u'_1}{u'_2} = \sum_{i,j\in[\ell]} u_1[i]u_1[j] u_2[i] u_2[j] + n^{4k+2}$.
\item for $v'_1,v'_2 \in V'$ coming from some $v_1,v_2 \in S$, $\ip{v'_1}{v'_2} = \sum_{i,j\in[\ell]} v_1[i]v_1[j] v_2[i] v_2[j] + n^{4k+2}$.
\end{itemize} Note that by our choice of $k$, $|u_1[i]u_1[j] u_2[i] u_2[j]| \leq n^{4k}$ for all $i,j$. So all inner products are positive, in both of the above cases. In contrast, for $u' \in U'$ and $v' \in V'$, \[\langle u',v'\rangle = \sum_{i,j \in [\ell]} u[i]u[j]v[i]v[j] = \sum_{i,j} u[i]v[i]\cdot u[j]v[j] = (\langle u,v\rangle)^2.\] Now, all possible inner products between every $u' \in U'$ and $v' \in V'$ are non-negative, and $\ip{u'}{v'}=0$ if and only if $\ip{u}{v}=0$. 

Suppose we normalize all vectors in $U'$ and $V'$, replacing each vector $u'$ and $v'$ by $u'' := u'/||u'||_2$. Since \[\langle u'',v''\rangle = \frac{1}{||u'||_2\cdot ||v'||_2}\langle u,v\rangle,\]
the vector pairs in $U',V'$ with zero inner product are exactly preserved, and all inner products of pairs within $U'$ (and of pairs within $V'$) are still positive. By the law of cosines, for all $u'' \in U'$ and $v'' \in V'$ we have
\[||u''-v''||^2_2 =  ||u''||^2_2 - ||v'||^2_2 - 2\langle u'',v''\rangle = 2 - 2\langle u'',v''\rangle.\] 
Therefore, taking $S := U' \cup V'$, solving Hopcroft's problem on $S$ is equivalent to finding two vectors with $\ell_2$-distance at least $\sqrt{2}$, and this is the maximum possible distance between two vectors in the instance. It follows that solving $\ell_2$-furthest pair on these instances will solve Hopcroft's problem on them as well.
\end{proof}

\begin{corollary}[Hardness of Bichromatic $\ell_2$-Closest Pair] \label{closest-pair} Under SETH (or OVC), finding a bichromatic $\ell_2$-closest pair in $\omega(\log \log n)^2$ dimensions requires $n^{2-o(1)}$ time, with vectors of $O(\log n)$-bit entries.
\end{corollary}

\begin{proof} 
As before, we begin from the proof of hardness for Hopcroft's problem (Theorem~\ref{hopcroft}). The reduction there computes $O(n^{\eps})$ instances of Hopcroft's problem on $n$ points in $\Omega(1/\eps \cdot c \log(c) \cdot \log \log n)$ dimensions, for any desired $\eps > 0$. We will slightly modify the proof of Theorem~\ref{furthest-pair} for furthest pair to work for bichromatic closest pair.

Let $S$ be a set of $n$ vectors in $\Z^{\ell}$ with $\ell = \omega(\log \log n)$ and $O(\log n)$-bit entries. We wish to know if two vectors in $S$ are orthogonal. 

Let $v[i]$ denote the $i$th component of a vector $v$. We define the vectors in $U'$ very similarly to the proof of Theorem~\ref{furthest-pair}: for all $u \in S$, make the $\ell^2$-dimensional vector
\[u' := [u[1]\cdot u[1], u[1]\cdot u[2], \ldots, u[i]\cdot u[j], \ldots, u[\ell]\cdot u[\ell]].\] That is, each component of $u'$ is a product of two components of $u$. Put each $u'$ in a set $U'$ of \emph{red points}. For every vector $v \in S$, define the $\ell^2$-dimensional vector
\[v' := [-v[1]\cdot v[1], -v[1] \cdot v[2], \ldots, -v[i]\cdot v[j], \ldots, -v[\ell]\cdot v[\ell]],\] and put $v'$ in a set $V'$ of \emph{blue points}. 
Now observe that for every red $u' \in U'$ and every blue $v' \in V'$, 
\[\ip{u'}{v'} = -(\langle u,v\rangle)^2.\] Thus the inner product between a red $u'$ and a blue $v'$ is zero when $\ip{u}{v} = 0$, and is otherwise negative. If we normalize all vectors in $U'$ and $V'$, those red-blue pairs with zero inner product are preserved, and the rest of the red-blue pairs still have negative inner product. Analogously as in Theorem~\ref{furthest-pair}, this means that the red-blue pairs with zero inner product have Euclidean distance $\sqrt{2}$, and all other red-blue pairs have distance strictly greater than $\sqrt{2}$. Therefore finding the closest red-blue pair in this $\ell^2$-dimensional instance will solve the original instance $S$ of Hopcroft's problem.
\end{proof}

\input{all-nearest-hardness}

\section{Conclusion}

We have given some rigorous explanation for why certain point-location and proximity problems only admit barely-subquadratic time algorithms: they can encode difficult high-dimensional Boolean problems in surprisingly low dimensions. In contrast, the nearly-linear proximity problems seem incapable of such an encoding; moreover, if any of them \emph{were} found to be capable, we would be refuting some major conjectures in fine-grained complexity. 

It is likely that many more consequences can be derived than what we have shown here. 
\begin{itemize}
\item For one example, Backurs and Indyk (personal communication) have noticed that our lower bound for bichromatic $\ell_2$-Closest Pair implies an inapproximability result for the \emph{fast Gauss transform}~\cite{Greengard-Strain91}, where we are given a set of $n$ red vectors $R$ and $n$ blue vectors $B$ in $\R^d$, and are asked to compute \[F(r) = \sum_{b \in B}	e^{-||a-b||^2}\] for every $r \in R$. In particular, they have observed that (under OVC) $F$ cannot be approximated with an additive $\eps$-error in $n^{2-\delta} \cdot \poly(\log(1/\eps),2^d)$ time, for any fixed $\delta > 0$. 
\item For another example, a variant of the reduction in Lemma~\ref{dim-red} (where instead of setting $x := d+1$ in the polynomials, we imagine trying \emph{all} choices for $x$ from a large-enough field, and we build larger-dimensional Boolean vectors whose inner products model the process of computing inner products among all values of $x$) was used in recent work with Abboud and Rubenstein~\cite{ARW17} to show that finding a vector pair of maximum inner product among a set of $n$ Boolean $n^{o(1)}$-dimensional vectors is hard to non-trivially approximate in sub-quadratic time.
\end{itemize}

There are many interesting questions to pursue further; here are some particularly compelling ones.
\begin{enumerate}
\item Can the $\omega(\log \log n)$ and $\omega(\log \log n)^2$ dimensionality in our hardness reductions be reduced, all the way down to $\omega(1)$ dimensions? This would demonstrate very tight hardness for solving these problems. The main bottleneck is that in the main reduction (Theorem~\ref{REDDIM} and Lemma~\ref{dim-red}) it seems we have to compute $O(\log n)^{O(\log n)/\ell}$ different instances to go from $O(\log n)$ dimensions down to $\ell$ dimensions; perhaps there is a more efficient reduction method.

\item All of the nearly-linear problems discussed in this paper actually have $2^{O(d)} \cdot n \log^{O(1)} n$-time algorithms, except for bichromatic $\ell_1$ and $\ell_{\infty}$ closest pair, for which their best known algorithms have the running time bound $n \cdot \log^{O(d)} n$. Could stronger hardness be established for these two problems, or can their dependence on $d$ be improved? So far, prior work~\cite{Williams05,David16} has only established quadratic-time hardness for these problems when $d =\omega(\log n)$, so it is quite possible that they are in fact solvable in $2^{O(d)} \cdot  n \log^{O(1)} n$ time, like the other nearly-linear problems.

\item The All-Nearest Neighbors problem is solvable in $2^{O(d)} \cdot n\log^{O(1)} n$ time in the general case, not just when all vectors are in $\{-1,0,1\}$. Is there is a dimensionality reduction for the special case of $\{-1,0,1\}$, similar to Lemma~\ref{dim-red}? (Please note that this would likely refute OVC and SETH.)

\item Do any of the ``popular conjectures'' in fine-grained complexity imply that $\ell_2$-Closest Pair requires $n^{2-o(1)}$ time in $\omega(\log n)$ dimensions?
\end{enumerate}

\section*{Acknowledgements}

I am grateful to Amir Abboud, Arturs Backurs, Piotr Indyk, Aviad Rubenstein, and Huacheng Yu for useful comments and discussions. In particular, several years ago Huacheng took patient notes from one of our meetings, and wrote up a version of the main lemma presented here. Unfortunately he declined to be an author on this paper.

\bibliographystyle{alpha}
\bibliography{papers}

\end{document}

%% file: all-nearest-hardness.tex
\section{Hardness of Euclidean All-Nearest Neighbors}

Here we prove hardness for All-Nearest Neighbors in $\omega(\log n)$ dimensions: 

\begin{reminder}{Theorem~\ref{all-nn}} Under OVC, the All-Nearest Neighbors problem in $\ell^{\omega(\log n)}_2$ requires $n^{2-o(1)}$ time, even restricted to vectors with entries from $\{-1,0,1\}$.
\end{reminder}

It seems plausible that there is a sub-quadratic-time reduction from All-Nearest Neighbors to $\ell_2$-Closest Pair (even in high dimensions), so we think of Theorem~\ref{all-nn} as good evidence that $\ell_2$-Closest Pair is also hard for $\omega(\log n)$ dimensions.

\begin{proof} Let $d = O(\log n)$. We begin with the Subset Containment problem: \emph{Given $n$ red subsets of $[d]$ and $n$ blue subsets of $[d]$, is there some red subset that is contained in some blue subset?} It is well-known that this problem is equivalent to \ov{} on $n$ vectors in $d$ dimensions~\cite{Williams05} (imagine you have a red/blue version of \ov{} with vectors in $\{0,1\}^d$, and flip all the bits of the blue vectors; this converts the \ov{} instance to a Subset Containment instance). 

The main idea of the proof is to use error correcting codes over $\{-1,1\}$ to keep the red points ``far apart'' from each other, so that the nearest neighbor of each red point $x$ is a blue point $y$ which is as close to being a superset of $x$ as possible. 

Let $R$ be the collection of red sets and let $B$ be the blue sets. We will think of them as vectors in $\{0,1\}^d$ in the natural way. First, we do a trick which will help control the vector norms. Try all pairs of integers $d_1, d_2 \in [d]$ with $d_1 < d_2$. Take the subset $R_{d_1}$ of $R$ which only contains vectors having exactly $d_1$ ones, and take the subset $B_{d_2}$ of $B$ which only contains vectors having exactly $d_2$ ones. We will work with the collections $R' := R_{d_1}$ and $B' := B_{d_2}$ in the following. (The benefit is that we now may assume that all red vectors have the same norm value $v_A$, and all blue vectors have the same norm value $v_B$, and it only costs $O(d^2)$ extra calls.)

Let $\eps \in (0,1/2)$. We say that a \emph{code with distance at least $(1/2-\eps)$} is a collection of vectors $S \subseteq \{-1,1\}^k$ such that for all $u, v \in S$ with $u \neq v$, $\ip{u}{v} \leq 2\eps k$. Note this condition is equivalent to saying that the Hamming distance between each pair of $k$-dimensional vectors is at least $(1/2-\eps)k$. Such codes are known to have polynomial-time constructions. In particular, it was recently shown how to efficiently construct a set $S$ of at least $n$ such vectors, with dimension only $k \leq t=O(\log n)/\eps^{2+o(1)}$~\cite{TaShma17}. In the following, let $S$ be such a code with $\eps = 1/8$ and the dimension $k$ as a parameter to be set later. 

We will add $k$ dimensions to all vectors in $R'$ and $B'$. For each vector $v_i \in R$, for $i=1,\ldots,n$, we concatenate the $i$th codeword from $S$ to the end of it, obtaining a $(d+k)$-dimensional vector $v'_i$. For each vector $w_i \in B$, we concatenate $k$ zeroes to the end, obtaining a $(d+k)$-dimensional $w'_i$.

Observe that for all vectors $v'_i$ from $R'$, their $\ell_2$-norm squared is 
\begin{align}\label{red-norm-squared}||v'_i||^2_2 = d_1 + \sum_{i=1}^{k} 2^2 = d_1 + 4k\end{align} For all vectors $w'_i$ from $B'$, we have $||w'_i||^2_2 = d_2$.
Furthermore, observe that for every two vectors $v'_i,v'_j$ from $R'$, their inner product is at most $(d_1-1)+2\eps k$, because the original vectors $v_i$ and $v_j$ were distinct vectors with exactly $d_1$ ones (so their inner product is at most $d_1-1$), and the inner product of any two distinct codewords is at most $2\eps k$. Therefore we have 
\begin{align*}
||v'_i-v'_j||^2_2 &= ||v'_i||^2_2 + ||v'_j||^2_2 - 2\ip{v'_i}{v'_j}\\
& = 2(d_1 + 4k) - 2\ip{v'_i}{v'_j} \text{~~~~~~(by \eqref{red-norm-squared})}\\
& \geq 2(d_1 + 4k) - 2(d_1-1 + 2\eps k) = 8k + 2 - 4\eps k = (8-1/2)k + 2.\end{align*} 
On the other hand, for a vector $v'_i$ from $R'$ and a vector $w'_j$ from $B'$, 
\[||v'_i-w'_j||^2_2 = ||v'_i||^2_2 + ||w'_j||^2_2 - 2\ip{v'_i}{w'_j} = d_1 + 4k + d_2 - 2\ip{v'_i}{w'_j}.\] 
Note that the inner product $\ip{v'_i}{w'_j}$ is maximized when the original subset $v_i$ (of cardinality $d_1$) is contained in the subset $w_j$ (of cardinality $d_2$), in which case $\ip{v'_i}{w'_j} = d_1$. So the minimum possible distance between $v'_i$ and $w'_j$ is  
\[||v'_i-w'_j||^2_2 = d_1 + 4k + d_2 - 2\ip{v'_i}{w'_j}  = (d_2 - d_1) + 4k.\] 
Putting it all together, suppose we set $k$ large enough that \[(8-1/2)k + 2 > d+4k\] (e.g. $k \geq d$ will do). From there, if there is some red set (of cardinality $d_1$) in $R$ contained in a blue set (of cardinality $d_2$) in $B$, then the nearest neighbor of the corresponding point in $R'$ will be a point in $B'$ with distance precisely $(d_2 - d_1) + 4k$ from it. Set $k = \Theta(\log n)$ so that it is at least $d$, and it is large enough to support at least $n$ distinct codewords with $\eps = 1/8$. 

We have reduced \ov{} with $n$ vectors in $\{0,1\}^{c \log n}$ to $n$ points in $\{-1,0,1\}^{\Theta(\log n)}$, such that computing all-nearest neighbors in $\ell_2$ will determine if the original instance had a red set contained in a blue set. In particular, we can check for every point whether its nearest neighbor corresponds to a set containing it in the original instance, or a set it contains. By the above, there is a red set contained in a blue set if and only if for the cardinalities $d_1$ and $d_2$ of these respective sets, the nearest neighbor to some point $v$ in $R_{d_1}$ is a point in $B_{d_2}$ with distance only $(d_2 - d_1) + 4k$ from $v$.
\end{proof}